\documentclass[letterpaper,11pt]{article}
\usepackage{amsmath,amssymb,amsfonts,amsthm}
\usepackage{setspace}
\usepackage{bbm}
\usepackage{color}
\usepackage{enumerate}
\usepackage{mathrsfs}
\usepackage[longnamesfirst]{natbib}
\usepackage[bottom]{footmisc}
\usepackage{graphicx}
\usepackage{subcaption}
\usepackage{hyperref}
\usepackage{cleveref}
\usepackage[shortlabels]{enumitem}
\usepackage[usenames,dvipsnames,svgnames,table]{xcolor}
\usepackage{multirow}
\usepackage[margin=1.2in]{geometry}
\usepackage{enumitem}
\usepackage{appendix}
\usepackage{authblk}

\newtheorem{theorem}{Theorem}
\newtheorem{lemma}{Lemma}
\newtheorem{proposition}{Proposition}
\newtheorem{assumption}{Assumption}

\newtheorem{corollary}{Corollary}

\newtheorem{remark}{Remark}[section]
\newtheorem{example}{Example}

\crefname{assumption}{Assumption}{Assumptions}
\crefname{subassumptioni}{Assumption}{Assumptions}
\crefname{section}{Appendix}{Appendices}

\newlist{subassumption}{enumerate}{1}
\setlist[subassumption,1]{
  label=(\roman*),   
  ref=\theassumption(\roman*), 
  leftmargin=2em
}

\newcommand{\E}{\operatorname{\mathbb{E}}}

\newcommand{\B}{\mathbb{B}}
\newcommand{\R}{\mathbb{R}}
\newcommand{\G}{\mathbb{G}}

\newcommand{\weakto}{\rightsquigarrow}

\DeclareMathOperator*{\argmax}{argmax}
\DeclareMathOperator*{\argmin}{argmin}

\hypersetup{
     colorlinks   = true,
     citecolor    = blue!50!black,
     urlcolor     = blue!50!black,
     linkcolor    = blue!50!black
}

\title{\textbf{
     Generalized Method of Moments\\
     with Partially Missing Data}\footnote{
     We appreciate helpful comments and suggestions from Andres Santos. 
     Research reported in this publication was supported by the National Institute on Aging of the National Institutes of Health and in part by the Social Security Administration under Award Number U01AG077280. The content is solely the responsibility of the authors and does not necessarily represent the official views of the National Institutes of Health.
}
}

\date{\today}

\author[1]{\textsc{Grigory Franguridi}}
\author[2]{\textsc{Hyungsik Roger Moon}}

\affil[1]{Center for Economic and Social Research, University of Southern California  \vspace{1ex}}
\affil[2]{Department of Economics, University of Southern California}

\begin{document}

\maketitle

\begin{abstract}  
\linespread{1.2}

We consider a generalized method of moments framework in which a part of the data vector is missing for some units in a completely unrestricted, potentially endogenous way.
In this setup, the parameters of interest are usually only partially identified.
We characterize the identified set for such parameters using the support function of the convex set of moment predictions consistent with the data.
This identified set is sharp, valid for both continuous and discrete data, and straightforward to estimate.
We also propose a statistic for testing hypotheses and constructing confidence regions for the true parameter, show that standard nonparametric bootstrap may not be valid, and suggest a fix using the bootstrap for directionally differentiable functionals of \citet{fang2019inference}.
A set of Monte Carlo simulations demonstrates that both our estimator and the confidence region perform well when samples are moderately large and the data have bounded supports.

\medskip 

\noindent \textbf{Keywords:} endogenous attrition, partial identification, panel data, support function

\medskip

\noindent \textbf{JEL codes:} C14, C23
\end{abstract}

\newpage

\section{Introduction}

Missing data are ubiquitous in empirical research, particularly in repeated-measurement settings. \citet{rubin1976inference} classified mechanisms of missingness into three categories: missing completely at random (MCAR), missing at random (MAR), and missing not at random (MNAR). In this paper, we consider the last case, where missingness may be correlated with the data in an arbitrary manner.  
More specifically, we consider a generalized method of moments framework in which a portion of the data vector is missing for some units in a completely unrestricted, potentially endogenous way.
A representative example is a two-period panel where, in the second period, some units drop out of the sample in a way that is correlated with both outcomes and covariates in both periods.

When no restriction on missingness is imposed, the parameters of interest are usually only partially identified.
We characterize the identified set for the parameter of interest via the support function of the set of moments consistent with the observed data. We show that the identified set is sharp, is valid for both continuous and discrete data, and is straightforward to estimate.

We also show how to perform hypothesis testing and construct confidence regions using the estimate of the minimum of the support function as a test statistic. We derive the limit distribution of this test statistic and provide a procedure for computing critical values, following \citet{fang2019inference}. We also show that the test controls size locally uniformly using an abstract result in \citet{fang2019inference}.

Partial identification with missing data has been widely studied since Manski's seminal work, \citet{Manski1989anatomy} (see, e.g., \citet{Manski2005partial} and \citet{Molinari2020microeconometrics} for a survey). Our paper contributes to this extensive literature, with its key feature being that the identified set of moment predictions is convex, and the support function of this set can be estimated using a simple sample analog. 

Our model is neither a moment inequality model with a convex identified set (e.g., \citet{KaidoSantos2014}), nor the intersection bounds model of \citet{ChernozhukovLeeRosen2013intersection}. It is similar to the model with convex moment predictions as in \citet{BeresteanuMolchanovMolinari2011}, the main difference being that our characterization of the sharp identified set does not rely on a representation via random closed sets.

As mentioned above, the primary application of our model is a panel regression with attrition that may be endogenous. The analysis of panel data with attrition has a long history, and numerous techniques for handling attrition under various assumptions have been proposed. These include modelling the attrition process parametrically (e.g., \citet{HausmanWise1979}), using auxiliary information in the form of refreshment samples (e.g., \citet{hirano2001combining}), imposing semiparametric assumptions on the attrition process (e.g., \citet{bhattacharya2008inference}), using inverse probability weighting (e.g., \citet{Wooldridge2002}), selection models with unobserved heterogeneity (e.g., \citet{SemykinaWooldridge2010}), and multiple imputation and Bayesian methods (e.g., \citet{DengEtAlSurvey}).

The rest of the paper is organized as follows.
Section \ref{sec:framework} introduces the general framework and characterizes the sharp identified set.
Section \ref{sec:estimation} suggests an estimator of the identified set.
Section \ref{sec:inference} develops a bootstrap procedure for hypothesis testing.
Section \ref{sec:simulation} explores the performance of our methodology in a set of Monte Carlo simulations.
Section \ref{sec:conclusion} concludes.
All proofs are given in the Appendix.

\section{Setup and the identified set}\label{sec:framework}

We consider the generalized method of moments framework when some of the observed variables may be missing for some observations.
Specifically, let $Z_{1i}$ be a data vector that we observe for each unit $i=1,\dots,n$ of the sample and let $Z_{2i}$ be a data vector that may be missing for some units.
Let $S_i\in \{0,1\}$ be the sample selection indicator, which is equal to 1 when $Z_{2i}$ is observed and 0 otherwise.
Dropping the subscript $i$, the observed vector is then $W=(S,Z_1',S Z_2')'$.
A key aspect of our framework is that we allow arbitrary dependence of $S$ on both observed and unobserved features.
In particular, $S$ can be arbitrarily correlated with $Z_2$, representing endogenous sample selection.
We do not allow for overlap between elements of $Z_1$ and $Z_{2}$.
More generally, no features of the distribution of $Z_2$ are assumed to be identified from the random sample of $Z_1$.
We are interested in a finite-dimensional parameter $\theta\in\Theta\subset \R^{d_\theta}$ that is defined by general moment conditions
\begin{align}
    \E_{\pi}[\phi(Z_1,Z_2,\theta)]=0, \label{eq:GMM}
\end{align}
where $\phi$ is a $d_\phi$-dimensional moment function and $\pi$ is the joint distribution of $Z_1,Z_2$.
Of course, $\pi$ is not point-identified due to the presence of missing data (see below), and hence neither is $\theta$.

One special case of this framework is two-period panel data with unrestricted attrition in the second period.
To see this, let $Z_t=(X_t',Y_t) \in \R^d$ be the stacked vector of covariates and outcomes in period $t=1,2.$
In period 1, a random sample from $Z_1$ is observed.
In period 2, there is attrition, and so $Z_2$ is only observed when $S=1$, where $S\in\{0,1\}$ is the indicator of staying in the sample.
When the attrition is fully unrestricted and the parameter of interest is defined by a set of moment conditions \eqref{eq:GMM} (such as, for instance, the slope coefficient in the fixed effects linear regression), this model becomes a special case of the GMM with missing data described above.
For expositional clarity, we use the terminology related to this panel model throughout the rest of the paper.

We impose the following assumptions.

\begin{assumption}\label{a:moment-and-support}
    \begin{subassumption}
        \item \label{a:Theta-compact} The parameter space $\Theta$ is compact.
        \item \label{a:mu} The distribution $\pi$ of $(Z_1,Z_2)$ belongs to the set $\mathcal{P}(\mathcal{Z}_1\times \mathcal{Z}_2)$ of probability distributions on $\mathcal{Z}_1\times \mathcal{Z}_2$, where $\mathcal{Z}_1, \mathcal{Z}_2 \subset \R^d$ are known compact sets.
        \item \label{a:phi-continuous} For each $\theta\in\Theta$, the map $(z_1,z_2) \mapsto \phi(z_1,z_2,\theta)$ is continuous.
        \item \label{a:E-phi-continuous} For each $\pi \in \mathcal{P}(\mathcal{Z}_1\times \mathcal{Z}_2)$, the map $\theta \mapsto \E_{\pi} [\phi(Z_1,Z_2,\theta)]$ is continuous.
        
    \end{subassumption}
\end{assumption}

\Cref{a:Theta-compact,a:E-phi-continuous,a:phi-continuous} are standard and impose compactness of the parameter space and the continuity of moments.
\Cref{a:mu} posits that the researcher knows the compact set $\mathcal{Z}_1 \times \mathcal{Z}_2$ that contains (but does not have to be equal to) the true support of the data.
We impose no restrictions on the data-generating process beyond these basic requirements.

To characterize the identified set for $\theta$, let $p=P(S=1) \in (0,1)$ be the unconditional selection probability and notice that $\pi$ can be decomposed into the point-identified distribution of $(Z_1,Z_2)|S=1$, denoted by $\pi^1$, and the partially identified distribution of $(Z_1,Z_2)|S=0$, denoted by $\pi^0$, viz.,
\begin{align*}
    \pi = p \pi^1 + (1-p)\pi^0.
\end{align*}
Denote by $\pi_1^0$ the point-identified probability distribution of $Z_1|S=0$.
The sharp identified set for $\pi^0$ is the convex set of the distributions satisfying \Cref{a:mu} with first marginal $\pi_1^0$, i.e.,
\begin{align*}
    \Pi^0 = \left\{\pi^0 \in \mathcal{P}(\mathcal{Z}_1\times \mathcal{Z}_2): \,\, \pi^0(A \times \mathcal Z_2) = \pi_1^0(A) \text{ for all Borel sets }A \subset \mathcal{Z}_1 \right\}.
\end{align*}
Therefore, the sharp identified set for $\pi$ is
\begin{align*}
    \Pi = \left\{p \pi^1 + (1-p) \pi^0, \,\, \pi^0 \in \Pi^0 \right\}.
\end{align*}
Write
\begin{align}
     \E_{\pi}[\phi(Z_1,Z_2,\theta)] = p \E_{\pi^1}[\phi(Z_1,Z_2,\theta)] + (1-p) \E_{\pi^0} [\phi(Z_1,Z_2,\theta)] =: \nu_{\pi^0}(\theta).
\end{align}
Then the sharp identified set for $\theta$ is
\begin{align*}
    \Theta_I = \left\{\theta\in\Theta: \,\, \exists \pi^0 \in \Pi^0 \text{ such that } \nu_{\pi^0}(\theta)=0 \right\}.
\end{align*}
The next proposition establishes a convenient characterization of $\Theta_I$ in terms of the support function of the compact, convex set\footnote{See Step 1 in \Cref{app:id-set} for the proof of convexity and compactness of $N(\theta)$.} of moment predictions
\begin{align*}
    N(\theta) := \left\{\nu_{\pi^0}(\theta):\,\,\pi^0\in\Pi^0 \right\}.
\end{align*}

\begin{proposition}[identified set via support function]\label{prop:id-set}
    Suppose \Cref{a:moment-and-support} holds and let
    \begin{align*}
        \psi_{N(\theta)}(u) = \max_{\nu\in N(\theta)} u'\nu
    \end{align*}
    be the support function of $N(\theta)$. Define the criterion function
    \begin{align*}
        Q(\theta) = \min_{u\in \B} \psi_{N(\theta)}(u),
    \end{align*}
    where $\B$ is the unit ball in $\R^{d_\phi}$.
    Then $\Theta_I$ is compact and given by
    \begin{align}
        \Theta_I = \left\{ \theta\in\Theta: \,\, Q(\theta) = 0 \right\}.
    \end{align}
\end{proposition}

\begin{proof}
    See \Cref{app:id-set}.
\end{proof}

\begin{remark}
    In the above characterization, the Euclidean unit ball ~$\B$ can be replaced by any compact set containing the origin in its interior, such as the unit ball in the $\ell^1$ norm or the uniform norm. This may turn out to be convenient computationally.
\end{remark}

\begin{remark}
    A similar strategy was recently employed by \citet{franguridi2025inference} for characterizing the identified set in a partially identified GMM framework based on optimal transport.
\end{remark}

The support function $\psi_{N(\theta)}$ can be written as 
\begin{align*}
    \psi_{N(\theta)}(u) &= \max_{\pi\in \Pi} u'\E_\pi[\phi(Z_1,Z_2,\theta)]\\
    & = p \E[u'\phi(Z_1,Z_2,\theta)|S=1] + (1-p) \max_{\pi^0\in\Pi^0} \E_{\pi^0}[u'\phi(Z_1,Z_2,\theta]].
\end{align*}
The last term involves optimization over an infinite-dimensional set of distributions. Fortunately, we can convert it into a finite-dimensional program for every value of $Z_1$, as the following proposition shows.

\begin{proposition}[reduction formula] \label{prop:reduction} Suppose \Cref{a:moment-and-support} holds. Then
    \begin{align*}
        \max_{\pi^0\in\Pi^0} \E_{\pi^0}[u'\phi(Z_1,Z_2,\theta)] = \E\left[\max_{z_2\in\mathcal{Z}_2} u'\phi(Z_1,z_2,\theta) |S=0 \right].
    \end{align*}
\end{proposition}

\begin{proof}
    See \Cref{app:reduction}.
\end{proof}

\begin{figure}
    \centering
    \includegraphics[width=0.8\textwidth]{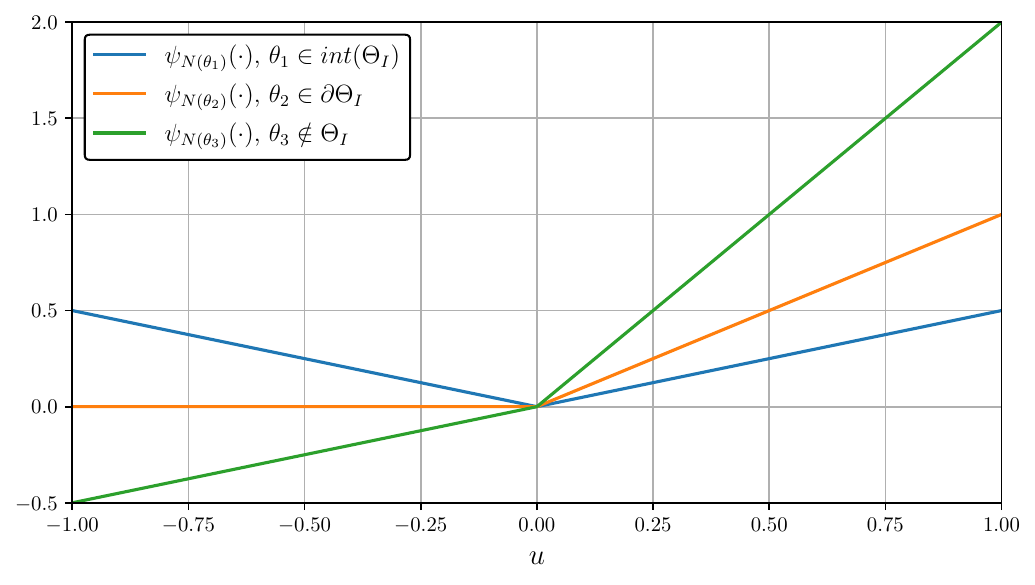}
    \caption{Set identification via support functions}
    \label{fig:univariate}
\end{figure}

\begin{example}\label{example:univariate}
Suppose there is only one moment condition, i.e., $\phi$ is scalar-valued.
In this case, the set of moment predictions is the closed interval 
\begin{align*}
    N(\theta) = \left[ \min_{\pi\in\Pi} \E_\pi[\phi(Z_1,Z_2,\theta)],\,\, \max_{\pi\in\Pi} \E_\pi[\phi(Z_1,Z_2,\theta)] \right],
\end{align*}
The support function of $N(\theta)$ is a piecewise linear function given by
\begin{align*}
    \psi_{N(\theta)}(u) = \max_{\pi\in\Pi} u \E_\pi[\phi(Z_1,Z_2,\theta)] = 
    \begin{cases}
        u \cdot \min_{\pi\in\Pi} \E_\pi[\phi(Z_1,Z_2,\theta)], &\text{if } u < 0,\\
        u \cdot \max_{\pi\in\Pi} \E_\pi[\phi(Z_1,Z_2,\theta)], &\text{if } u \ge 0,
    \end{cases}
\end{align*}
and hence the criterion function is 
\begin{align*}
    Q(\theta) = \min_{u\in \B} \psi_{N(\theta)}(u) = \min\left\{0,\,\, \max_{\pi\in\Pi} \E_\pi[\phi(Z_1,Z_2,\theta)], \,\,- \min_{\pi\in\Pi} \E_\pi[\phi(Z_1,Z_2,\theta)] \right\}.
\end{align*}
The identification condition $Q(\theta) = 0$ can be easily seen to be equivalent to
\begin{align*}
    \min_{\pi\in\Pi} \E_\pi[\phi(Z_1,Z_2,\theta)] \le 0 \le \max_{\pi\in\Pi} \E_\pi[\phi(Z_1,Z_2,\theta)],
\end{align*}
which, of course, is just the condition that zero belongs to the set of moment predictions.
See \Cref{fig:univariate} for an illustration of how the argmax set of $\psi_{N(\theta)}$ determines the position of $\theta$ relative to the identified set $\Theta_I$.
Notice also that using the reduction formula in \Cref{prop:reduction}, we can write
\begin{align*}
    \min_{\pi\in\Pi} \E_\pi[\phi(Z_1,Z_2,\theta)] = \E[f_{\min}(W,\theta)], \\
    \max_{\pi\in\Pi} \E_\pi[\phi(Z_1,Z_2,\theta)] = \E[f_{\max}(W,\theta)],
\end{align*}
where $W=(S,Z_1',S Z_2')'$ is the data vector and 
\begin{align*}
    f_{\min}(W,\theta) = S \phi(Z_1,S Z_2,\theta) + (1-S) \min_{z_2\in\mathcal{Z}_2 } \phi(Z_1,z_2,\theta), \\
    f_{\max}(W,\theta) = S \phi(Z_1,S Z_2,\theta) + (1-S) \max_{z_2\in\mathcal{Z}_2} \phi(Z_1,z_2,\theta).
\end{align*}

\end{example}
\section{Estimation}\label{sec:estimation}

Combining \Cref{prop:id-set,prop:reduction} suggests the following estimator of the identified set:
\begin{align*}
    \hat\Theta_I = \left\{ \theta\in\Theta: \,\, \hat Q(\theta) \ge - \eta_n \right\},
\end{align*}
where $\eta_n>0$ is a tuning parameter and the sample criterion function
\begin{align*}
    \hat Q(\theta) = \min_{u\in\B} \hat \psi_{N(\theta)}(u) = \frac{1}{n} \sum_{i=1}^n \left( s_i u'\phi(z_{i1}, z_{i2},\theta) + (1-s_i) \max_{z_2\in\mathcal{Z}_2} u'\phi(z_{i1}, z_{2},\theta) \right).
\end{align*}
Define the distance from a point $x \in \R^{d_\phi}$ to a set $A \subset \R^{d_\phi}$ by
\begin{align*}
    d(x,A) = \inf_{a\in A} \|a-x\|.
\end{align*}
For any sets $A,B \subset \R^{d_{\theta}}$, define the Hausdorff distance
\begin{align*}
    d_H(A,B) = \max \left\{ \sup_{a\in A} d(a,B), \,\, \sup_{b\in B} d(b,A)  \right\}.
\end{align*}
We impose the following assumptions.

\begin{assumption}\label{as:separation}
    There exists a weakly increasing function $m:\R_+ \to \R_+$ such that $m(0)=0$, $m(\delta)>0$ for $\delta>0$ and $\min_{u\in \B} \psi_{N(\theta)}(u) \le - m(d(\theta,\Theta_I))$.
\end{assumption}

\begin{assumption}\label{as:eta}
    $\eta_n \downarrow 0$ and $n^{-1/2} = o(\eta_n)$.
\end{assumption}

\Cref{as:separation} imposes separation of the identified set $\Theta_I$ in terms of the criterion function $Q(\theta) = \min_{u\in \B} \psi_{N(\theta)}(u)$.
A sufficient condition can be formulated in terms of the gradient of $Q(\theta)$ away from the identified set.
\Cref{as:eta} states that $\eta_n$ has to converge to zero slower than $O(n^{-1/2})$.

The following proposition establishes the consistency of our estimator in the Hausdorff distance.

\begin{theorem}[Hausdorff consistency]\label{prop:consistency}
    Under \Cref{a:moment-and-support,as:separation,as:eta}, $d_H(\Theta_I,\hat\Theta_I)=o_p(1)$.
\end{theorem}

\begin{proof}
    See \Cref{app:consistency}.
\end{proof}
\section{Inference}\label{sec:inference}

We now wish to develop a test of the hypothesis $H_0:\,\theta=\theta_0$.
We suggest using the statistic
\begin{align*}
    \hat T(\theta_0) = \sqrt{n} \min_{u\in\B} \hat \psi_{N(\theta_0)}(u),
\end{align*}
which is the estimator of the scaled negative distance to the identified set of moment predictions $N(\theta_0)$.
Note that more negative values of $\hat T(\theta_0)$ indicate stronger evidence against the null.
To establish the asymptotic distribution of this statistic, we first prove that the function $\hat \psi_{N(\theta)}(u)$ converges to a Gaussian process uniformly over $u\in\B$ and $\theta\in\Theta$.

\begin{assumption}\label{a:phi-finite-variance}
    $\E\left[ \sup_{\theta\in\Theta}\sup_{z_2\in\mathcal{Z}_2} \|\phi(Z_1,z_2,\theta)\|^2 \right]<\infty$.
\end{assumption}

\begin{assumption}\label{a:phi-Lipschitz}
    There exists a random variable $L(Z_1)$ such that $\E L(Z_1)^2 < \infty$ and 
    \begin{align*}
        \sup_{z_2\in\mathcal{Z}_2} \|\phi(Z_1,z_2,\theta_1)-\phi(Z_1,z_2,\theta_2)\| \le L(Z_1)\|\theta_1-\theta_2\|
    \end{align*}
    for all $\theta_1,\theta_2\in\Theta$, almost surely in $Z_1$.
\end{assumption}

\begin{theorem}[uniform CLT for the support function]\label{prop:uclt}
    Suppose \Cref{a:moment-and-support,a:phi-finite-variance,a:phi-Lipschitz} hold. Then there exists a Gaussian process $\G$ with trajectories in $C(\B\times\Theta)$ such that
    \begin{align*}
        \sqrt n (\hat \psi_{N(\theta)}(u) - \psi_{N(\theta)}(u)) \weakto \G(u,\theta).
    \end{align*}
\end{theorem}

\begin{proof}
    See \Cref{app:uclt}.
\end{proof}

\begin{remark}
    The results of \citet{KaidoSantos2014} suggest that the estimator $\hat\psi_{N(\theta)}(u)$ is asymptotically semiparametrically efficient for $\psi_{N(\theta)}(u)$. 
\end{remark}

An application of the functional delta method to \Cref{prop:uclt} yields the following result for the asymptotic distribution of the test statistic $\hat T(\theta)$ over $\theta\in\Theta_I$.

\begin{corollary}[asymptotic distribution of test statistic]\label{prop:lim-distribution}
    Under the assumptions of \Cref{prop:uclt}, 
    \begin{align*}
        \hat T(\theta) \weakto \min_{u\in U(\theta)} \G(u,\theta)    
    \end{align*}
    uniformly over $\theta\in\Theta_I$, where
    \begin{align*}
        U(\theta) = \argmin_{u\in \B} \psi_{N(\theta)}(u).
    \end{align*}
\end{corollary}

\begin{proof}
    See \Cref{app:lim-distribution}.
\end{proof}

\begin{remark}
    Using this uniform convergence result, it should be possible to develop a specification test based on the statistic $\max_{\theta\in\Theta} \hat T(\theta)$, the population analog of which is zero if and only if the identified set $\Theta_I$ is nonempty. 
    The critical values can be obtained using the bootstrap of \citet{fang2019inference}, in analogy to our treatment of the test statistic $\hat T(\theta_0)$ for a fixed $\theta_0$ below.
    Similar specification tests have been developed for moment inequality models, see, e.g., \citet{bugni2015specification} and references therein.
\end{remark}

Since the asymptotic distribution of $\hat T(\theta)$ is neither available in closed form nor is easy to simulate from, we need an alternative procedure for obtaining the critical values.
The results of \citet{fang2019inference} imply that, unless the Hadamard derivative $\min_{u \in U(\theta)} \G(u,\theta)$ is a linear functional of $\G$, the standard nonparametric bootstrap is invalid.

\setcounter{example}{0}
\begin{example}[continued]
    In the univariate case, the sample criterion function is 
\begin{align*}
    \hat Q(\theta) = \min_{u\in \B} \hat \psi_{N(\theta)}(u) = \min\left\{0,\,\, \hat\E[f_{\max}(W,\theta)], \,\,- \hat \E[f_{\min}(W,\theta)] \right\}.
\end{align*}
Consider a data generating process and a parameter value $\theta$ satisfying $\E[f_{\min}(W,\theta)] < 0$, making the sample moment $\hat \E[f_{\min}(W,\theta)]$ irrelevant asymptotically.
We have
\begin{align*}
    Q(\theta) &= \min\left\{0,\,\, \E[f_{\max}(W,\theta)] \right\}, \\
    \hat Q(\theta) &\overset{asy}{\sim} \min\left\{0,\,\, \hat\E[f_{\max}(W,\theta)] \right\}.
\end{align*}
This is known to be an irregular problem when $\E[f_{\max}(W,\theta)]=0$, which corresponds to $\theta$ being on the boundary of the identified set.
In particular, the standard nonparametric bootstrap is invalid for the asymptotic distribution $\min(0,N(0,1))$ of $\sqrt{n} \hat Q(\theta)$, see, e.g., \citet{andrews2000inconsistency}.
\end{example}

We show how to leverage the bootstrap for directionally differentiable functionals of \citet{fang2019inference}.
To this end, we propose an estimator of the derivative that satisfies their key Assumption 4.

For $\varepsilon_n>0$, define 
\begin{align}
    \hat\chi'(h) = \min_{u\in \hat U(\varepsilon_n)} h(u), \label{eq:hdd-estimator}
\end{align}
where
\begin{align*}
    \hat U(\varepsilon_n) = \left\{u\in \B:\,\,\hat\psi_{N(\theta_0)}(u) \le \min_{v\in\B} \hat\psi_{N(\theta_0)}(v)+\varepsilon_n \right\}
\end{align*}
is the ``$\varepsilon_n$-enlarged'' argmin set of $\hat\psi_{N(\theta_0)}$.
We impose the following assumptions.

\begin{assumption}[sharp minima of $\psi_{N(\theta_0)}$]\label{as:sharp-minima}
    There exists $\kappa>0$ such that 
    \begin{align*}
        \psi_{N(\theta_0)}(u) \ge \min_{v\in \B}\psi_{N(\theta_0)}(v) + \kappa \cdot d(u,U(\theta_0)) \text{ for all } u \in \B.
    \end{align*}
\end{assumption}

\begin{assumption}[bandwidth]\label{as:bandwidth} The sequence $\varepsilon_n$ satisfies $\varepsilon_n \downarrow 0$ and $\|\hat\psi_{N(\theta_0)}-\psi_{N(\theta_0)}\|_\B = o_p(\varepsilon_n)$.
\end{assumption}

\Cref{as:sharp-minima} posits that $\psi_{N(\theta_0)}$ grows sufficiently fast away from its argmin set $U(\theta_0)$.
The latter does not have to be a singleton, as illustrated in \Cref{example:univariate} for which \Cref{as:sharp-minima} holds.
This assumption is equivalent to every subgradient $\nabla\psi_{N(\theta_0)}(u)$ being bounded away from zero for $u\notin U(\theta_0)$.
It also suffices that it holds in a small neighborhood around $U(\theta_0)$ rather than on the entire set $\B$.
\Cref{as:bandwidth} restricts how fast $\varepsilon_n$ should converge to zero.
Since $\|\hat\psi_{N(\theta_0)}-\psi_{N(\theta_0)}\|_\B=O_p(1/\sqrt{n})$ by \Cref{prop:uclt}, we can take $\varepsilon_n \sim \log n/\sqrt{n}$.

Given our estimator \eqref{eq:hdd-estimator}, the \citet{fang2019inference} bootstrap algorithm for testing $H_0:\theta=\theta_0$ at the nominal size $\alpha\in (0,1)$ is as follows.
\begin{enumerate}
    \item Compute $\hat T(\theta_0)=\sqrt{n} \min_{u\in\B} \hat \psi_{N(\theta_0)}(u)$.
    \item Define $\hat U_n = \left\{u\in \B:\,\, \hat \psi_{N(\theta_0)}(u) \le \min_{v\in\B} \hat \psi_{N(\theta_0)}(v) + \varepsilon_n \right\}$.
    \item For each $b=1,\dots, B$:
    \begin{enumerate}
      \item Draw a random sample $W_i^{b}$, $i=1,\dots, n$, with replacement from the data $W_i = (S_i,Z_{1i},S_i Z_{2i})$, $i=1,\dots, n$.
      \item Compute $\hat \psi_{N(\theta_0)}^{*b}(u)$ on the sample $W_i^b$, $i=1,\dots,n$.
      \item Compute
      \begin{align*}
        \hat T^{*b} = \min_{u\in\hat U_n}\sqrt{n}\left( \hat \psi_{N(\theta_0)}^{*b}(u) - \hat \psi_{N(\theta_0)}(u) \right).
      \end{align*}
    \end{enumerate}
  \item Denote by $\hat c_{\alpha}^*$ the $\alpha$-quantile of $\hat T^{*1},\dots, \hat T^{*B}$.
  \item Reject $H_0$ if $\hat T(\theta_0) < \hat c_{\alpha}^*$.
  \end{enumerate}

The following theorem establishes that this bootstrap-based testing procedure controls size for a fixed data generating process.

\begin{theorem}[test validity, pointwise]\label{thm:bootstrap}
    Suppose \Cref{as:sharp-minima,as:bandwidth} hold and consider the testing procedure above with $B=\infty$. Then for every fixed data generating process, the asymptotic probability of rejecting the true null hypothesis is less than or equal to $\alpha$.
\end{theorem}

\begin{proof}
    See \Cref{app:bootstrap}.
\end{proof}

Although the previous theorem establishes that our test asymptotically controls size for a fixed data generating process, it is silent about uniform size control.
In other words, it could be possible that even for a large sample size, one can find a particularly unfavorable DGP that would lead our test to exhibit significant size distortion.
Such lack of uniform size control has been extensively studied in the context of moment inequality models, e.g., \citet{imbens2004confidence,romano2008inference,andrews2009validity,moon2009estimation,woutersen2006simple}.

Fortunately, we can rely on the abstract results of \citet{fang2019inference} to establish local size control.
Denote by $P$ the joint distribution of the data vector $W=(S,Z_1',S Z_2')'$ and let $\psi(P)$ be the corresponding population support function of $N(\theta_0)$, where $\theta_0$ is any value in the identified set $\Theta_I = \Theta_I(P)$. 
We consider a sequence of DGPs $P_t$, $t \ge 0$, that is local to a given DGP $P_0$ in the following sense, cf. Assumption 5 in \citet{fang2019inference}.

\begin{assumption}[local sequence of DGPs]\label{as:local-analysis}
    The path $t \mapsto P_t$ is quadratic mean differentiable and satisfies for any $\lambda >0$ the following.
    \begin{subassumption}
    \item \label{suba:psi-diff} There exists $\psi'(\lambda)$ such that
    \[
        \|\sqrt{n}(\psi(P_{\lambda/\sqrt{n}})-\psi(P_0)) - \psi'(\lambda)\|_\B \to 0.
    \]
    \item \label{suba:psi-regular} $\sqrt{n} (\hat\psi_n - \psi(P_{\lambda/\sqrt{n}})) \overset{\lambda}{\to} \G_0$, where $\overset{\lambda}{\to}$ denotes convergence in distribution under $\{W_i\}_{i=1}^n$ i.i.d., with each $W_i$ distributed as $P_{\lambda/\sqrt{n}}$, and $\hat\psi_n$ is our estimator of $\psi(P_{\lambda/\sqrt{n}})$ based on $\{W_i\}_{i=1}^n$.
    \item \label{suba:local-distribution} $\G_0$ is tight and supported on $C(\B)$.
    \end{subassumption}
\end{assumption}

\Cref{suba:psi-diff} imposes differentiability of the parameter $\psi(P)$ along the path $P_{\lambda/\sqrt{n}}$.
It is expected to hold for any quadratic mean differentiable path in our model under mild regularity conditions.
\Cref{suba:psi-regular} requires that the distribution of the estimator $\hat\psi_n$ is unaffected by local perturbations of the DGP, i.e. that $\hat\psi_n$ is \emph{regular}.
Finally, \Cref{suba:local-distribution} is expected to hold due to \Cref{prop:uclt}.

\begin{theorem}[test validity, locally uniform]\label{thm:uniform-test-validity}
    Suppose \Cref{as:sharp-minima,as:bandwidth,as:local-analysis} hold and consider the testing procedure above with $B=\infty$. Then this procedure controls size locally uniformly in the sense that
    \begin{align*}
        \limsup_{n \to \infty} P_{\lambda/\sqrt{n}}\left(\hat T_n(\theta_0) < \hat c_{\alpha}^* \right) \le \alpha.
    \end{align*}
\end{theorem}

\begin{proof}
    See \Cref{app:uniform-test-validity}.
\end{proof}

\section{Extension to more general missingness patterns}\label{sec:multi-period}

Our methodology can be extended to more general missingness patterns, as long as missingness is ``monotone'' (see below) and data contains no information on the distribution of the missing variables (i.e., missingness is ``fully unrestricted'').
We will illustrate this extension using the setup of panel data with more than two periods, where attrition happens monotonically, i.e., once the units drop out of the sample, they never return.
For expositional simplicity, we consider the case of $T=3$ periods.

Denote the data in the three periods by $Z_1,Z_2,Z_3$ and denote the indicators of staying in the sample in periods 2 and 3 by $S_2$ and $S_3$, respectively.
Since there is no return to the sample, $S_2=0$ implies $S_3=0$.
Also, let
\[
p_2=P(S_2=1), \quad p_3=P(S_3=1), \quad p_{2|0} = p(S_2=1|S_3=0).
\]
The parameter of interest $\theta$ is defined by the moment conditions
\[
\E_\pi [\phi(Z_1,Z_2,Z_3,\theta)]=0,
\]
where the expectation is taken with respect to the latent, partially identified distribution $\pi$ of $(Z_1,Z_2,Z_3)$.

\begin{proposition}\label{prop:multi-period}
The identified set is $\Theta_I = \{\theta\in\Theta: \,\, \min_{u\in\B} \psi_\theta(u)=0 \}$, where
\begin{align*}
    \psi_\theta(u) &= p_3\E[\phi(Z_1,Z_2,Z_3,\theta)|S_3=1] + (1-p_3)\left[p_{2|0} \E [\max_{z_3} u'\phi(Z_1,Z_2,z_3,\theta) |S_2=1,S_3=0] \right. \\
    &+ \left. (1-p_{2|0}) \E [\max_{z_2,z_3} u'\phi(Z_1,z_2,z_3,\theta) |S_2=S_3=0]
    \right].
\end{align*}
\end{proposition}

\begin{proof}
    See \Cref{app:multi-period}.
\end{proof}

\section{Monte Carlo simulation}\label{sec:simulation}

We consider a simple cross-sectional linear regression
\begin{align*}
    Y_i = \theta_0 + \theta_1 X_i + \varepsilon_i, \quad i=1,\dots,n,
\end{align*}
where $Y_i$ is observed for the entire sample, whereas $X_i$ is missing for observations $i$ with $S_i=0$, where $S_i \in \{0,1\}$ is the sample selection indicator.
The parameter of interest is $\theta=(\theta_0,\theta_1)$.
We assume that the support of $X_i$ is known, in agreement with \Cref{a:mu}.
Dropping the subscript $i$, the moment conditions are
\begin{align*}
    \E(Y-\theta_0-\theta_1 X)&=0, \\
    \E[X(Y-\theta_0-\theta_1 X)]&=0,
\end{align*}
and hence the support function is
\begin{align*}
    \psi_\theta(u_1,u_2) &= p \E[u_1 (Y-\theta_0-\theta_1 X) + u_2 X(Y-\theta_0-\theta_1 X) \,|\,S=1  ] \\
    &+ (1-p)\E\left[\max_{x\in[0,1]}\left\{ u_1 (Y-\theta_0-\theta_1 x) + u_2 x(Y-\theta_0-\theta_1 x) \right\}\right].
\end{align*}
Notice that the inner maximization is a constrained quadratic problem.
We assume that $X_i$ has a uniform[0,1] distribution and $\varepsilon_i$ has a uniform[-1,1] distribution independently of $X_i$. Selection is completely random with probability $p=P(S_i=1)=0.9$.
We set the sample size $n=1000$ and the tuning parameters $\eta_n=\varepsilon_n = 0.1\log n/\sqrt{n} \approx 0.022$.

\begin{figure}
    \centering
    \includegraphics[width=1\textwidth]{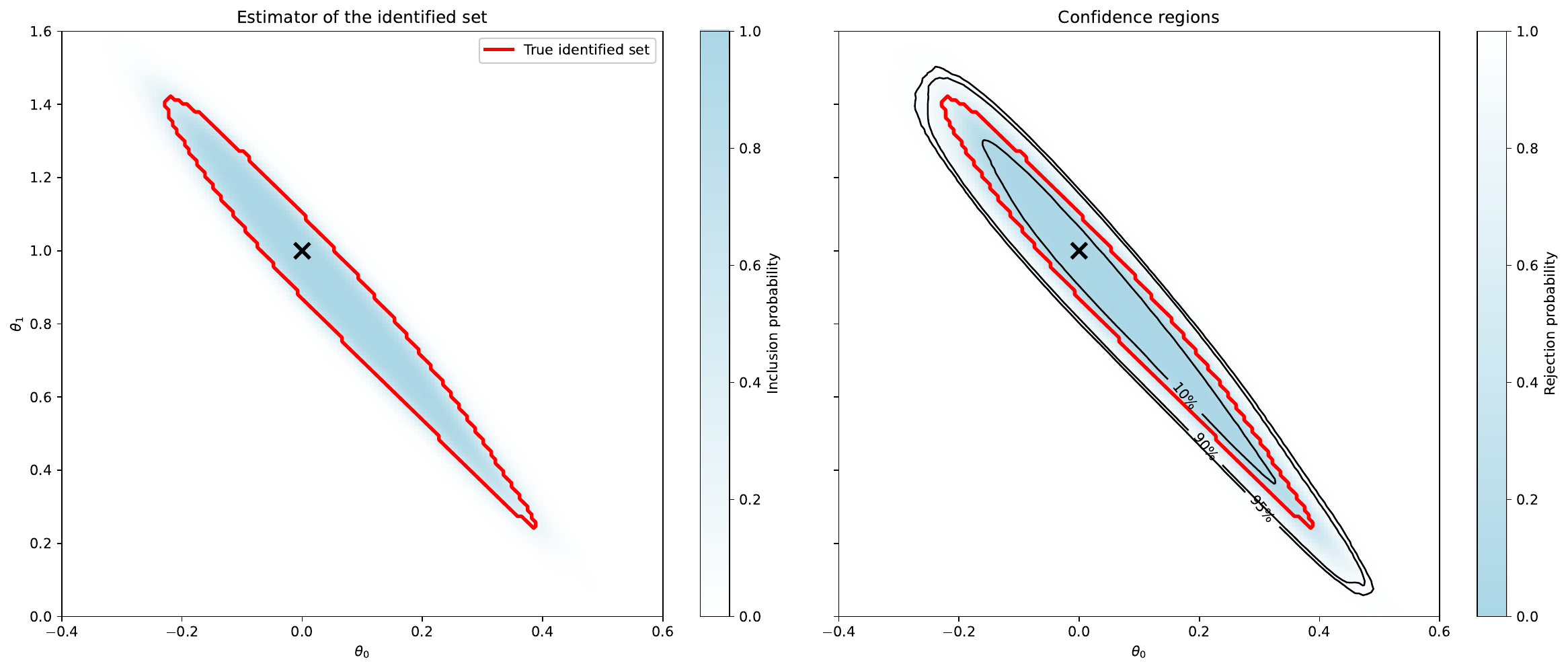}
    \caption{Estimates of the identified set and confidence regions with nominal sizes $1-\alpha \in \{0.10,0.90,0.95\}$ averaged over $1000$ simulations.}
    \label{fig:mc}
\end{figure}

The simulation results are illustrated in \Cref{fig:mc}. 
The true identified set (red contour line) is approximated using a sample of $10,000$ observations.
Our estimates of the identified set track the true identified set closely.
The confidence regions (three black contour lines) collect the hypothesized parameter values that are not rejected by the bootstrap test described in Section \ref{sec:inference} with the size $\alpha \in \{0.05,0.10,0.90\}$ and the number of bootstrap samples $B=1000$.
They are found to be quite narrow even at high nominal confidence levels $1-\alpha$.
\section{Conclusion}\label{sec:conclusion}

In this paper, we study the partial identification of parameters defined by GMM moment conditions when components of the data vector are missing in a completely unrestricted manner. A leading motivation is panel data subject to endogenous attrition.

We characterize the sharp identified set through the support function of the convex set of moment predictions consistent with the observed data.
For inference, we employ the minimum of the sample analog of the support function as a test statistic and develop a valid inference procedure, drawing on the bootstrap for directionally differentiable functionals of \citet{fang2019inference}.
We demonstrate that our estimator and confidence regions perform well in a set of Monte Carlo simulations.

\bibliographystyle{ecta}
\singlespacing
\bibliography{ref}

@article{HausmanWise1979,
	author = {Hausman, Jerry A. and Wise, David},
	title = {Attrition Bias in Experimental and Panel Data: The Gary Income Maintenance Experiment},
	journal = {Econometrica},
	year = {1979},
	volume = {47},
	number = {2},
	pages = {455--474}
}

@article{DengEtAlSurvey,
	author = {Deng, Y. and Chang, C. and others},
	title = {Statistical Methods for Nonignorable Missing Longitudinal Data},
	journal = {Statistical Science},
	year = {2016}
}

@book{Wooldridge2002,
	author = {Wooldridge, Jeffrey},
	title = {Econometric Analysis of Cross Section and Panel Data},
	publisher = {MIT Press},
	year = {2002}
}

@article{SemykinaWooldridge2010,
	author = {Semykina, Anastasia and Wooldridge, Jeffrey},
	title = {Estimating Panel Data Models in the Presence of Endogeneity and Selection: Theory and Application},
	journal = {Journal of Econometrics},
	year = {2010},
	volume = {157},
	pages = {375--380}
}

@article{Molinari2020microeconometrics,
	title={Microeconometrics with partial identification},
	author={Molinari, Francesca},
	journal={Handbook of econometrics},
	volume={7},
	pages={355--486},
	year={2020},
	publisher={Elsevier}
}

@article{ChernozhukovLeeRosen2013intersection,
	title={Intersection bounds: Estimation and inference},
	author={Chernozhukov, Victor and Lee, Sokbae and Rosen, Adam M},
	journal={Econometrica},
	volume={81},
	number={2},
	pages={667--737},
	year={2013},
	publisher={Wiley Online Library}
}

@article{Manski1989anatomy,
  title={Anatomy of the selection problem},
  author={Manski, Charles F},
  journal={Journal of Human resources},
  pages={343--360},
  year={1989},
  publisher={JSTOR}
}

@article{Manski2005partial,
  title={Partial identification with missing data: concepts and findings},
  author={Manski, Charles F},
  journal={International Journal of Approximate Reasoning},
  volume={39},
  number={2-3},
  pages={151--165},
  year={2005},
  publisher={Elsevier}
}

@article{rubin1976inference,
  title={Inference and missing data},
  author={Rubin, Donald B},
  journal={Biometrika},
  volume={63},
  number={3},
  pages={581--592},
  year={1976},
  publisher={Oxford University Press}
}

@article{hirano2001combining,
  title={Combining panel data sets with attrition and refreshment samples},
  author={Hirano, Keisuke and Imbens, Guido and Ridder, Geert and Rubin, Donald B},
  journal={Econometrica},
  volume={69},
  number={6},
  pages={1645--1659},
  year={2001}
}

@article{bhattacharya2008inference,
  title={Inference in panel data models under attrition caused by unobservables},
  author={Bhattacharya, Debopam},
  journal={Journal of Econometrics},
  volume={144},
  number={2},
  pages={430--446},
  year={2008},
  publisher={Elsevier}
}

@article{fang2019inference,
  title={Inference on directionally differentiable functions},
  author={Fang, Zheng and Santos, Andres},
  journal={The Review of Economic Studies},
  volume={86},
  number={1},
  pages={377--412},
  year={2019},
  publisher={Oxford University Press}
}

@article{franguridi2025inference,
  title={Inference in partially identified moment models via regularized optimal transport},
  author={Franguridi, Grigory and Liu, Laura},
  journal={Working paper},
  year={2025},
}

@article{shapiro1991asymptotic,
  title={Asymptotic analysis of stochastic programs},
  author={Shapiro, Alexander},
  journal={Annals of Operations Research},
  volume={30},
  number={1},
  pages={169--186},
  year={1991},
  publisher={Springer}
}

@book{vaart2023empirical,
  title={Weak Convergence and Empirical Processes: With Applications to Statistics},
  author={van der Vaart, A. and Wellner, Jon},
  year={2023},
  publisher={Springer}
}

@article{andrews2000inconsistency,
  title={Inconsistency of the bootstrap when a parameter is on the boundary of the parameter space},
  author={Andrews, Donald WK},
  journal={Econometrica},
  pages={399--405},
  year={2000},
  publisher={JSTOR}
}

@article{andrews2009validity,
  title={Validity of subsampling and “plug-in asymptotic” inference for parameters defined by moment inequalities},
  author={Andrews, Donald WK and Guggenberger, Patrik},
  journal={Econometric Theory},
  volume={25},
  number={3},
  pages={669--709},
  year={2009},
  publisher={Cambridge University Press}
}

@article{imbens2004confidence,
  title={Confidence intervals for partially identified parameters},
  author={Imbens, Guido W and Manski, Charles F},
  journal={Econometrica},
  volume={72},
  number={6},
  pages={1845--1857},
  year={2004},
  publisher={Wiley Online Library}
}

@article{moon2009estimation,
  title={Estimation with overidentifying inequality moment conditions},
  author={Moon, Hyungsik Roger and Schorfheide, Frank},
  journal={Journal of Econometrics},
  volume={153},
  number={2},
  pages={136--154},
  year={2009},
  publisher={Elsevier}
}

@article{romano2008inference,
  title={Inference for identifiable parameters in partially identified econometric models},
  author={Romano, Joseph P and Shaikh, Azeem M},
  journal={Journal of Statistical Planning and Inference},
  volume={138},
  number={9},
  pages={2786--2807},
  year={2008},
  publisher={Elsevier}
}

@article{woutersen2006simple,
  title={A simple way to calculate confidence intervals for partially identified parameters},
  author={Woutersen, Tiemen},
  journal={Unpublished Manuscript, Johns Hopkins University},
  year={2006}
}

@article{BeresteanuMolchanovMolinari2011,
  title={Sharp identification regions in models with convex moment predictions},
  author={Beresteanu, Arie and Molchanov, Ilya and Molinari, Francesca},
  journal={Econometrica},
  volume={79},
  number={6},
  pages={1785--1821},
  year={2011}
}

@article{KaidoSantos2014,
  title={Asymptotically efficient estimation of models defined by convex moment inequalities},
  author={Kaido, Hiroaki and Santos, Andres},
  journal={Econometrica},
  volume={82},
  number={1},
  pages={387--413},
  year={2014}
}

@book{aliprantis2006infinite,
  title={Infinite dimensional analysis: a hitchhiker’s guide},
  author={Aliprantis, Charalambos D and Border, Kim C},
  year={2006},
  publisher={Springer}
}

@article{bugni2015specification,
  title={Specification tests for partially identified models defined by moment inequalities},
  author={Bugni, Federico A and Canay, Ivan A and Shi, Xiaoxia},
  journal={Journal of Econometrics},
  volume={185},
  number={1},
  pages={259--282},
  year={2015},
  publisher={Elsevier}
}

@book{dudley2018real,
  title={Real analysis and probability},
  author={Dudley, Richard M},
  year={2018},
  publisher={Chapman and Hall/CRC}
}

\onehalfspacing
\frenchspacing
\newpage
\appendix
\appendixpage

\section{Proof of \Cref{prop:id-set}}\label{app:id-set}

\textbf{Step 1: $N(\theta)$ is compact and convex}.

\noindent
Convexity is trivial.
Since $\mathcal{Z}_1\times\mathcal{Z}_2$ is compact, Theorem 15.11 and Corollary 15.7 in \citet{aliprantis2006infinite} imply that $\Pi^0$ is compact in the weak topology and the map $\pi^0 \mapsto \nu_{\pi^0}(\theta)$ is continuous.
Hence, $N(\theta)$ is compact as the image of a compact set under a continuous map.

\noindent
\textbf{Step 2: $\Theta_I = Q^{-1}(\{0\})$}.

\noindent
Since $N(\theta)$ is closed and convex, $0\in N(\theta)$ if and only if its support function is everywhere nonnegative: $\psi_{N(\theta)}(u) \ge 0$ for all $u\in \R^{\dim(\phi)}$.
Indeed, if $0\in N(\theta)$, then $\psi_{N(\theta)}(u) = \max_{\nu\in N(\theta)} u'\nu \ge 0$.
Conversely, if $0\notin N(\theta)$, then by strong separation of a point from the closed convex set $N(\theta)$, there exists $u\neq 0$ and $\alpha>0$ such that $u'\nu\le \alpha$ for all $\nu\in N(\theta)$.
This implies $\psi_{N(\theta)}(u) \le -\alpha < 0$.

Since $\psi_{N(\theta)}(0)=0$ for all $\theta$, nonnegativity of $\psi_{N(\theta)}$ is equivalent to the equality of
\begin{align}
    \min_{u\in \R^{\dim(\phi)}} \psi_{N(\theta)}(u) \label{eq:min-psi-Rd}
\end{align}
to zero.
Let us show that the latter is equivalent to the equality of
\begin{align}
    \min_{u\in \B} \psi_{N(\theta)}(u) \label{eq:min-psi-B}   
\end{align}
to zero.
Indeed, if \eqref{eq:min-psi-Rd} is zero, then the minimum is achieved at $u=0$, and hence \eqref{eq:min-psi-B} is zero.
Conversely, if \eqref{eq:min-psi-Rd} is nonzero, then there exists $u$ such that $\psi_{N(\theta)}(u) < 0$. By the positive homogeneity of the support function, $\psi_{N(\theta)}(u/\|u\|) = \psi_{N(\theta)}(u)/\|u\|<0$, and hence the expression \eqref{eq:min-psi-B} is negative.

\noindent
\textbf{Step 3: $\Theta_I$ is compact}.

\noindent
Since $\Theta_I = Q^{-1}(\{0\})$, it suffices to establish the continuity of $Q(\theta)$.
For this, notice that
\begin{align*}
    (u,\theta,\pi^0) \mapsto u'\nu_{\pi^0}(\theta)
\end{align*}
is a continuous function on the compact set $\B \times \Theta\times \Pi^0$, where $\Pi^0$ is equipped with the weak topology.
Berge's maximum theorem (see, e.g., Theorem 17.31 in \citet{aliprantis2006infinite}) implies that the function
\begin{align*}
    (u,\theta) \mapsto \max_{\pi^0\in\Pi^0}  u'\nu_{\pi^0}(\theta) = \psi_{N(\theta)}(u)
\end{align*}
is continuous on the compact set $\B \times \Theta$.
Applying Berge's theorem again implies that $Q(\theta)$ is continuous, completing the proof.

\section{Proof of \Cref{prop:reduction}}\label{app:reduction}

Denote $c(z_1,z_2)=u'\phi(z_1,z_2,\theta)$ and let $\Pi^0$ be the set of probability measures $\pi^0$ on $\mathcal Z_1\times\mathcal Z_2$ with the first marginal $\pi_1^0(\cdot)=\mathbb P(Z_1\in\cdot\mid S=0)$.

Fix any $\pi^0\in\Pi^0$. There exists a Markov kernel $\pi^0(dz_2\mid z_1)$ such that
\[
\pi^0(dz_1,dz_2)=\pi_1^0(dz_1)\,\pi^0(dz_2\mid z_1),
\]
see, e.g., Theorems 10.2.1 and 10.2.2 in \citet{dudley2018real}.
Hence
\[
\E_{\pi^0}[c(Z_1,Z_2)]
=
\int_{\mathcal Z_1}
\left(
\int_{\mathcal Z_2} c(z_1,z_2)\,\pi^0(dz_2\mid z_1)
\right)
\pi_1^0(dz_1).
\]
For each $z_1$,
\[
\int_{\mathcal Z_2} c(z_1,z_2)\,\pi^0(dz_2\mid z_1)
\le
\sup_{z_2\in\mathcal Z_2} c(z_1,z_2)
=
\max_{z_2\in\mathcal Z_2} c(z_1,z_2),
\]
where the equality follows from continuity of $c(z_1,\cdot)$ and compactness
of $\mathcal Z_2$.
Integrating and taking the maximum over $\pi^0\in\Pi^0$ yields
\begin{align}
    \max_{\pi^0\in\Pi^0} \E_{\pi^0}[c(Z_1,Z_2)]
\le \int \max_{z_2\in\mathcal Z_2} c(z_1,z_2)\, \pi_1^0(dz_1). \label{eq:ub}
\end{align}

We now establish a matching lower bound. Let $\{q(k)\}_{k\ge1}$ be a countable dense subset of $\mathcal Z_2$. For each $N\ge1$, define
\[
m_N(z_1):=\max_{1\le k\le N} c(z_1,q(k))
\]
and let $k_N(z_1)$ be an arbitrary element of the set
\[
\argmax_{1\le k\le N} c(z_1,q(k)) = \{k\in\{1,\dots,N\}: \,\, c(z_1,q(k)) = m_N(z_1)\}.
\]
Denote $q_N(z_1)=q(k_N(z_1))$ so that $m_N(z_1)=c(z_1,q_N(z_1))$. Let $\delta_{q_N(z_1)}(dz_2)$ be the point mass at $q_N(z_1)$. 
Define the joint distribution
\[
\pi_N(dz_1,dz_2):=\pi_1^0(dz_1)\,\delta_{q_N(z_1)}(dz_2).
\]
Then $\pi_N\in\Pi^0$ and
\[
\E_{\pi_N}[c(Z_1,Z_2)]
=
\int m_N(z_1)\,\pi_1^0(dz_1).
\]
Since $\{q(k)\}_{k\ge 1}$ is dense in the compact set $\mathcal{Z}_2$ and $c(z_1,\cdot)$ is continuous, we have, as $N\to \infty$, 
\[
m_N(z_1) = \max_{1\le k\le N} c(z_1,q(k)) \uparrow \max_{1\le k < \infty} c(z_1,q(k)) = \max_{z_2\in\mathcal Z_2} c(z_1,z_2) \quad\text{ for all } z_1 \in \mathcal{Z}_1.
\]
Moreover, $|m_N(z_1)|\le \sup_{z_2 \in \mathcal Z_2}|c(z_1,z_2)|\le C$ for some constant $C<\infty$ independent of $z_1$ due to continuity of $c(\cdot,\cdot)$ and compactness of $\mathcal Z_1$ and $\mathcal Z_2$. Hence, by dominated convergence,
\[
\lim_{N\to\infty} \int m_N(z_1)\,\pi_1^0(dz_1) = \int \max_{z_2\in\mathcal Z_2} c(z_1,z_2) \,\pi_1^0(dz_1).
\]
Therefore,
\[
\max_{\pi^0\in\Pi^0}\E_{\pi^0}[c(Z_1,Z_2)] \ge \sup_{N \ge 1} \E_{\pi_N}[c(Z_1,Z_2)]
= \int \max_{z_2\in\mathcal Z_2} c(z_1,z_2)\,\pi_1^0(dz_1).
\]
Combining with the upper bound \eqref{eq:ub} yields
\[
\max_{\pi^0\in\Pi^0} \E_{\pi^0}[c(Z_1,Z_2)] = \int \max_{z_2\in\mathcal Z_2} c(z_1,z_2)\, \pi_1^0(dz_1) = \E \left[ \max_{z_2\in\mathcal Z_2} c(Z_1,z_2) \,\middle|\, S=0 \right],
\]
which completes the proof.

\section{Proof of \Cref{prop:consistency}}\label{app:consistency}

We will follow the proof strategy in \citet{franguridi2025inference} and show that $\Theta_I \subset \hat\Theta_I$ w.p.a.\ 1 and that for every $\delta>0$, $\hat\Theta_I\subset \Theta_I^\delta$ w.p.a.\ 1, where $\Theta_I^\delta$ is the $\delta$-enlargement of $\Theta_I$, i.e.
    \[
      \Theta_I^\delta = \left\{\theta\in\Theta: \,\, d(\theta,\Theta_I) \le \delta \right\}.
    \]
    Compactness of $\Theta_I$ will then imply the required Hausdorff convergence.
  
    First, let us show that $\Theta_I \subset \hat\Theta_I$ w.p.a.\ 1.
    We have
    \begin{align*}
      \sup_{\theta\in\Theta_I} \hat Q_n(\theta) \le \sup_{\theta\in\Theta_I} Q(\theta) + \|\hat Q_n-Q\|_\infty =  \|\hat Q_n-Q\|_\infty.
    \end{align*}
    By \Cref{prop:uclt} and \Cref{as:eta}, the right-hand side is smaller than $\eta_n$ w.p.a.\ 1, and hence $\Theta_I \subset \hat\Theta_I$ w.p.a.\ 1.
  
    Now let us show that for every $\delta>0$, $\hat\Theta_I\subset \Theta_I^\delta$ w.p.a.\ 1.
    By \Cref{as:separation}, we have $\inf_{\theta\notin\Theta_I^\delta} Q(\theta) \ge m(\delta) > 0$.
    On the event $\|\hat Q_n-Q\|_\infty \le r_n$, we have
    \begin{align*}
      \inf_{\theta\notin\Theta_I^\delta} \hat Q_n(\theta) \ge \inf_{\theta\notin \Theta_I^\delta} Q(\theta) - \|\hat Q_n-Q\|_\infty \ge m(\delta) - r_n.
    \end{align*}
    Choose $n$ large enough so that $r_n \le \eta_n < m(\delta)/2$. Then 
    \begin{align*}
      \inf_{\theta\notin\Theta_I^\delta} \hat Q_n(\theta) \ge m(\delta)/2 > \eta_n,
    \end{align*}
    and hence no $\theta$ outside of $\Theta_I^\delta$ belongs to $\hat \Theta_I$. Therefore, $\hat\Theta_I\subset \Theta_I^\delta$ w.p.a.\ 1.

\section{Proof of \Cref{prop:uclt}}\label{app:uclt}

Define the data vector $W=(S,Z_1',S Z_2')'$ and let
\begin{align*}
    f_{u,\theta}(W) = S u' \phi(Z_1,S Z_2,\theta) + (1-S) \max_{z_2\in\mathcal Z_2} u' \phi(Z_1,z_2,\theta),
\end{align*}
so that $\hat\psi_{N(\theta)}(u)= \frac{1}{n} \sum_{i=1}^n f_{u,\theta}(W_i)$ and $\psi_{N(\theta)}(u)=\E f_{u,\theta}(W)$.
For brevity, drop the arguments $Z_1$ and $S Z_2$ in $\phi(Z_1,S Z_2,\theta)$ and $\max_{z_2}\phi(Z_1,z_2,\theta)$.
Write
\begin{align*}
    &f_{u_1,\theta_1}(W)-f_{u_2,\theta_2}(W) \\
    &= S(u_1'\phi(\theta_1) - u_2'\phi(\theta_2)) + (1-S) \left( \max_{z_2} u_1'\phi(z_2,\theta_1) - \max_{z_2} u_2'\phi(z_2,\theta_2) \right) \\
    &= S\left((u_1-u_2)'\phi(\theta_1) - u_2'(\phi(\theta_2)-\phi(\theta_1))\right) \\
    &+ (1-S)\left( \max_{z_2} u_1'\phi(z_2,\theta_1) - \max_{z_2} u_2'\phi(z_2,\theta_1) + \max_{z_2} u_2'\phi(z_2,\theta_1) - \max_{z_2} u_2'\phi(z_2,\theta_2) \right).
\end{align*}
By \Cref{a:phi-Lipschitz}, the term multiplied by $S$ can be bounded in absolute value by
\begin{align*}
    \|u_1-u_2\| \|\phi(Z_1,S Z_2,\theta_1)\| + \|u_2\| L(Z_1) \|\theta_1-\theta_2\|.
\end{align*}
For the first term multiplying $(1-S)$, we have
\begin{align*}
    \left|\max_{z_2} u_1'\phi(z_2,\theta_1) - \max_{z_2} u_2'\phi(z_2,\theta_1) \right| \le \max_{z_2} |(u_1-u_2)'\phi(z_2,\theta_1)| \le \|u_1-u_2\| \max_{z_2} \|\phi(Z_1,z_2,\theta_1)\|.
\end{align*}
Finally, for the second term multiplying $(1-S)$, we have
\begin{align*}
    &\left|\max_{z_2} u_2'\phi(z_2,\theta_1) - \max_{z_2} u_2'\phi(z_2,\theta_2) \right| \le \max_{z_2} |u_2'(\phi(z_2,\theta_1)-\phi(z_2,\theta_2))| \\
    &\le \|u_2\| \max_{z_2}\|\phi(z_2,\theta_1)-\phi(z_2,\theta_2)\| \le L(Z_1) \|\theta_1-\theta_2\|,
\end{align*}
where the last inequality follows from $\|u_2\|\le 1$ and \Cref{a:phi-Lipschitz}.

Combining the inequalities above with \Cref{a:phi-finite-variance} implies that the function class $\mathscr{F}=\{f_{u,\theta}:\,u\in \B,\,\theta\in\Theta\}$ is Lipschitz in the parameter with a square integrable constant.
By Theorem 2.7.17 in \citet{vaart2023empirical}, $\mathscr{F}$ is Donsker.

\section{Proof of \Cref{prop:lim-distribution}}\label{app:lim-distribution}

For any $\psi \in C(\B)$, the functional $\psi \mapsto \min_{u\in\B} \psi(u)$ is Hadamard directionally differentiable with the derivative
\begin{align*}
    h \mapsto \min_{u\in U(\psi)} h(u), \quad h\in C(\B),
\end{align*}
where
\[
U(\psi)=\arg\min_{u\in\B} \psi(u),
\]
as established in Theorem 3.1 in \citet{shapiro1991asymptotic}.
Applying the functional delta method (e.g., Theorem 3.10.5 in \citet{vaart2023empirical}) to \Cref{prop:uclt} completes the proof.

\section{Proof of \Cref{thm:bootstrap}}\label{app:bootstrap}

We need the following two technical lemmas, which are completely abstract and independent of our framework.

\begin{lemma}\label{lem:hausdorff-convergence}
    Suppose \Cref{as:sharp-minima,as:bandwidth} hold. Let $\hat\psi$ be random functions on $\B$ converging uniformly in probability to $\psi$. Then
    \[
    d_H(\hat U(\varepsilon_n), U) = o_p(1).
    \]
\end{lemma}
\begin{proof}
    The Hausdorff convergence is implied by
    \begin{align}
        \sup_{u\in\hat U(\varepsilon_n)} d_H(u, U) = o_p(1) \label{eq:outer-convergence}
    \end{align}
    and 
    \begin{align}
        \sup_{u\in U} d_H(u,\hat U(\varepsilon_n)) = o_p(1). \label{eq:inner-convergence}
    \end{align}
    First, let us establish \eqref{eq:outer-convergence}.
    Let $\delta_n=\|\hat\psi-\psi\|_\B$. For $v\in \hat U(\varepsilon_n)$, we have
    \begin{align*}
        \psi(v) \le \hat\psi(v)+\delta_n \le \min_{u\in\hat U(\varepsilon_n)}\hat\psi(u)  + \varepsilon_n + \delta_n \le  \min_{u\in U} \psi(u) + 2\delta_n + \varepsilon_n,
    \end{align*}
    where the second inequality holds by the definition of $\hat U(\varepsilon_n)$.
    Combining with Assumption \ref{as:sharp-minima} yields
    \begin{align*}
        \kappa \cdot d_H(v,U) \le \psi(v)- \min_{u\in U} \psi(u) \le 2\delta_n + \varepsilon_n.
    \end{align*}
    Taking the supremum over $v\in \hat U(\varepsilon_n)$ and noting that $2\delta_n +\varepsilon_n \to 0$ establishes \eqref{eq:outer-convergence}.

    Now let us establish \eqref{eq:inner-convergence}.
    Indeed, for any $u\in U$,
    \begin{align*}
        \hat\psi(u)\le \psi(u)+\delta_n =  \min_{u\in U} \psi(u) + \delta_n \le  \min_{u\in \hat U(\varepsilon_n)} \hat \psi(u) + 2\delta_n.
    \end{align*}
    Consider the event $E_n = \{\delta_n\le\varepsilon_n/2\}$. Using the inequality above, on this event,
    \begin{align*}
        \hat\psi(u)\le \min_{u\in \hat U(\varepsilon_n)} \hat \psi(u) + \varepsilon_n,
    \end{align*}
    which implies $u\in \hat U(\varepsilon_n)$, and hence $U \subset \hat U(\varepsilon_n)$. 
    By Assumption \ref{as:bandwidth}, $P(E_n)\to 1$, and so $U \subset \hat U(\varepsilon_n)$ w.p.a. 1.
    Therefore, $\sup_{u\in U} d_H(u,\hat U(\varepsilon_n)) = 0$ w.p.a. 1, establishing \eqref{eq:inner-convergence}.
\end{proof}

\begin{lemma}\label{lem:min-convergence}
    Suppose $d_H(\hat U,U)=o_p(1)$, where $\hat U,U$ are nonempty compact subsets of $\B$. Then for any continuous function $h$ on $\B$, 
    \[
    \min_{u\in\hat U} h(u) - \min_{u\in U} h(u) = o_p(1).
    \]
\end{lemma}

\begin{proof}
    The modulus of continuity of $h$ is
    \[
    \omega_h(\varepsilon) = \max_{u_1,u_2\in \B:\, \|u_1-u_2\| \le \varepsilon} |h(u_1)-h(u_2)|.
    \]
    Since $h$ is continuous and $\B$ is compact, $\omega_h(\varepsilon)\downarrow 0$ as $\varepsilon \downarrow 0$.

    Pick $u^*\in\arg\min_{u \in U} h(u)$. By definition of the Hausdorff distance, there exists $\hat u \in \hat U$ such that $\|\hat u-u^*\|\le d_H (\hat U,U)$. Then
    \begin{align*}
        \min_{u \in \hat U} h(u) \le h(\hat u) \le h(u^*) + \omega_h(d_H(\hat U,U)),
    \end{align*}
    which implies
    \begin{align*}
        \min_{u\in \hat U} h(u) - \min_{u\in  U} h(u) \le \omega_h(d_H(\hat U,U)).
    \end{align*}

    Similarly, pick any $\hat u^* \in \arg\min_{u \in \hat U} h(u)$. There exists $u \in U$ such that $\|\hat u^* - u\|\le d_H(\hat U,U)$. Then
    \begin{align*}
        \min_{u\in U} h(u) \le h(\hat u^*) + \omega_h(d_H(\hat U, U)),
    \end{align*}
    which implies
    \begin{align*}
        \min_{u\in U} h(u) - \min_{u\in \hat U} h(u) \le \omega_h(d_H(\hat U,U)).
    \end{align*}
    But since $\omega_h(\varepsilon)$ is continuous at $\varepsilon=0$ and $d_H(\hat U,U) = o_p(1)$, by the continuous mapping theorem, $\omega_h(d_H(\hat U,U)) = o_p(1)$, completing the proof.
\end{proof}

Let us show that the estimator $\hat\chi'$ is amenable to Remark 3.4 in \citet{fang2019inference}.
Indeed, let $\hat \Delta = \|h_1-h_2\|_{\hat U(\varepsilon_n)}$ (uniform norm on $\hat U(\varepsilon_n)$). Then $h_2(u)\ge h_1(u)-\hat \Delta$, and so 
\[
\min_{u\in \hat U(\varepsilon_n)} h_2(u) \ge \min_{u\in \hat U(\varepsilon_n)} h_1(u) - \hat \Delta.
\]
Swapping $h_1$ and $h_2$ yields the reverse inequality. Combining the two inequalities yields
\[
    \left| \hat\chi'(h_1)-\hat\chi'(h_2)\right| = \left| \min_{u\in \hat U(\varepsilon_n)} h_1(u) - \min_{u\in \hat U(\varepsilon_n)} h_2(u)\right| \le \hat \Delta \le \|h_1-h_2\|_B,
\]
i.e. $\hat \chi'$ is 1-Lipschitz.
Therefore, according to Remark 3.4, it suffices to establish pointwise consistency of $\hat\chi'(h)$, i.e., for any $h \in C(\B)$,
\begin{align*}
    \left|\hat\chi'(h)-\chi'(h) \right| = o_p(1).
\end{align*}
Indeed, \Cref{prop:uclt} implies that $\hat\psi$ converges uniformly in probability to $\psi$ on $\B$.
By \Cref{lem:hausdorff-convergence}, the argmin sets $\hat U(\varepsilon_n)$ converge to $U$ in the Hausdorff distance in probability. Using \Cref{lem:min-convergence} completes the proof.

\section{Proof of \Cref{prop:multi-period}}\label{app:multi-period}

Dropping the argument $\theta$ for brevity,
\begin{align*}
\E [\phi(Z_1,Z_2,Z_3)] &= p_3 \E[\phi(Z_1,Z_2,Z_3)|S_3=1] + (1-p_3)\E[\phi(Z_1,Z_2,Z_3)|S_3=0].
\end{align*}
Notice that $ \E[\phi(Z_1,Z_2,Z_3)|S_3=1]$ is point identified since $S_3=1$ implies $S_2=1$.
Then
\begin{align*}
    \E[\phi(Z_1,Z_2,Z_3)|S_3=0] &= p_{2|0} \E[\phi(Z_1,Z_2,Z_3)|S_2=1,S_3=0] \\
    &+(1-p_{2|0})\E[\phi(Z_1,Z_2,Z_3)|S_2=S_3=0].
\end{align*}
Arguing as in the proof of \Cref{prop:reduction} completes the proof.

\section{Proof of \Cref{thm:uniform-test-validity}}\label{app:uniform-test-validity}

We map our testing framework into Section 3.4.2 of \citet{fang2019inference} and then apply their Corollary 3.2. To this end, multiply our statistic by $-1$ to obtain
\begin{align*}
    T(\theta) &= \chi(\psi_\theta) = -\min_{u\in \B} \psi_{\theta}(u), \\
    \hat T(\theta) &= \chi(\hat \psi_\theta) = -\min_{u\in \B} \hat\psi_{\theta}(u),
\end{align*}
where
\begin{align*}
    \chi(h) = -\min_{u\in \B}h(u).
\end{align*}
This functional is Hadamard differentiable at $\psi_\theta$ with the derivative
\begin{align*}
    \chi_{\psi_\theta}'(h) = -\min_{u\in U} h(u),
\end{align*}
where $U = \arg\min_{u\in\B} \psi_\theta(u)$.
Notice that $T(\theta)=0$ if $\theta\in\Theta_I$ and $T(\theta) > 0$ otherwise.
Hence, we are testing the hypothesis $H_0: \chi(\psi_{\theta_0}) \le 0$ against $H_1: \chi(\psi_{\theta_0})>0$, as in equation (31) of \citet{fang2019inference}.

Let us show now that the functional $\chi'(h)$ is convex.
Indeed, for all $h_1,h_2$ and $u\in U$, we have
    \begin{align*}
         t h_1(u) + (1-t) h_2(u) \ge t \min_{u\in U} h_1(u) + (1-t) \min_{u\in U} h_2(u),
    \end{align*}
    and hence
    \begin{align*}
        -\chi'(t h_1 + (1-t) h_2) &= \min_{u\in U} \left\{t h_1(u) + (1-t) h_2(u) \right\} \\
        & \ge t \min_{u\in U} h_1(u) + (1-t) \min_{u\in U} h_2(u) = -t \chi'(h_1) - (1-t) \chi'(h_2).
    \end{align*}
    Multiplying this inequality by $-1$ establishes the convexity of $\chi'(h)$.
    Applying Corollary 3.2 in \citet{fang2019inference} completes the proof.

\end{document}